\documentclass[11pt,english]{article}
\usepackage{amsfonts}
\usepackage{amsmath}
\usepackage{geometry}
\usepackage{amsthm}
\usepackage{graphicx}
\usepackage{color}
\usepackage[T1]{fontenc}
\usepackage[latin9]{inputenc}

\usepackage{float}

\newtheorem{theorem}{Theorem}[section]

\newtheorem{proposition}[theorem]{Proposition}

\newtheorem{defn}[theorem]{Definition}

\numberwithin{equation}{section}
\numberwithin{figure}{section}
\theoremstyle{plain}
\newtheorem{thm}{\protect\theoremname}
  \theoremstyle{plain}
  
  \theoremstyle{plain}
  \newtheorem{cor}[thm]{\protect\corollaryname}

\makeatother

\usepackage{babel}
  \providecommand{\corollaryname}{Corollary}
  \providecommand{\lemmaname}{Lemma}
\providecommand{\theoremname}{Theorem}

\newcommand{\musiclab}{\textsc{MusicLab}}
\newcommand{\mo}{\textsc{M\&O}}
\begin{document}

\title{Measuring and Optimizing Cultural Markets}
\author{
A. Abeliuk$^1$
\and
G. Berbeglia$^2$
\and
M. Cebrian$^1$
\and
P. Van Hentenryck$^3$
}
\maketitle

\noindent
$^1$ The University of Melbourne and National ICT Australia. \\
$^2$ Melbourne Business School (The University of Melbourne) and National ICT Australia.\\
$^3$ The Australian National University and National ICT Australia. \\
\\
Corresponding author: Pascal Van Hentenryck (pvh@nicta.com.au) \\

\begin{abstract}
  Social influence has been shown to create significant
  unpredictability in cultural markets, providing one potential
  explanation why experts routinely fail at predicting commercial
  success of cultural products. To counteract the difficulty of making
  accurate predictions, ``measure and react'' strategies have been
  advocated but finding a concrete strategy that scales for very large
  markets has remained elusive so far. 

  Here we propose a ``measure and optimize'' strategy based on an
  optimization policy that uses product quality, appeal, and social
  influence to maximize expected profits in the market at each
  decision point. Our computational experiments show that our policy
  leverages social influence to produce significant performance
  benefits for the market, while our theoretical analysis proves that
  our policy outperforms in expectation any policy not displaying
  social information.

  Our results contrast with earlier work which focused on showing the
  unpredictability and inequalities created by social influence.  Not
  only do we show for the first time that dynamically showing
  consumers positive social information under our policy increases the
  expected performance of the seller in cultural markets. We also show
  that, in reasonable settings, our policy does not introduce
  significant unpredictability and identifies ``blockbusters''.

  Overall, these results shed new light on the nature of social
  influence and how it can be leveraged for the benefits of the
  market.
\end{abstract}

\newpage
\section{Introduction}

Prediction in cultural markets is a wicked problem: On the one hand,
experts routinely fail at predicting commercial success for items such
as movies, books, and
songs~\cite{bielby1994all,caves2000creative,de2011hollywood,hirsch1972processing,peterson1971entrepreneurship}
while, on the other hand, hit products are generally considered
qualitatively superior to the less successful
ones~\cite{watts2012everything}.  If those hit products are of
superior quality, why can't experts detect them a priori? As Duncan
Watts puts it, in cultural markets ``everything is obvious ... once
you know the answer,'' as experts find it extremely hard to benefit
from their specific knowledge~\cite{eliashberg1997film}.

In their seminal paper~\cite{salganik2006experimental}, Salganik,
Dodds, and Watts offer a satisfactory explanation for this paradox:
The impact of social influence on the behavior of individuals
(word-of-mouth), a process a priori invisible to experts, may distort
the quality perceived by the customers, making quality and popularity
out of sync and expert predictions less reliable. To investigate this
hypothesis experimentally, they created an artificial music market
called the \musiclab{}.  Participants in the \musiclab{} were
presented a list of unknown songs from unknown bands, each song being
described by its name and band. The participants were divided into two
groups exposed to two different experimental conditions: The {\em
  independent} condition and the {\em social influence} condition. In
the first group (independent condition), participants were provided
with no additional information about the songs. Each participant would
decide which song to listen to from a random list. After listening to
a song, the participant had the opportunity to download it. In the
second group (social influence condition), each participant was
provided with an additional information: The number of times the song
was downloaded by earlier participants. Moreover, these participants
were presented with a list ordered by the number of
downloads. Additionally, to investigate the impact of social
influence, participants in the second group were distributed in eight
``worlds'' evolving completely independently. In particular,
participants in one world had no visibility about the downloads and
the rankings in the other worlds.
%The goal of this
%multi-verse setting was to try to understand why different universes
%would reach different outcomes regarding the popularity of the songs.

The \musiclab{} provides an experimental testbed of measuring the
unpredictability of cultural markets. By observing the evolution of
different words given the same initial conditions, the \musiclab{}
provides unique insights on the impact of social influence and the
resulting unpredictability. In particular, Salganik et al suggested
that social influence contributes to both unpredictability and
inequality of success, with follow-up experiments confirming these
initial
findings~\cite{salganik2009web,salganik2008leading,Muchnik2013,vandeRijt2014}
(see also \cite{hedstrom2006experimental} for a perspective on the
\musiclab{} paper).  These experimental results received further
theoretical support by the development of a generative model that
reproduces the experimental data \cite{krumme2012quantifying}. The
\musiclab{} model, which is presented below, teased out the role of
the quality and appeal of the songs and position bias in the
unpredictability and inequality of the experimental market.

Together these results painted a bleak scenario for cultural
markets. From the standpoint of the market, e.g., producers and
consumers of cultural products, this is an inefficient situation
\cite{shiller1992market,shleifer2000inefficient,jensen1978some}: Fewer
cultural products are consumed overall and some high-quality products
may go unnoticed, while lower-quality ones are oversold. It is thus an
important open question to determine whether this situation can be
remedied and how. In \cite{watts2012everything}, Watts advocates the
use of ``measure and react'' strategies to address the difficulty in
making correct predictions: ``Rather than predicting how people will
behave and attempting to design ways to make customers behave in a
particular way [...] we can instead measure directly how they respond
to a whole range of possibilities and react accordingly''
\cite{watts2012everything}. The \musiclab{} provides a unique setting
to validate the benefits of a ``measure and react'' strategy.

This paper presents a ``measure and optimize'' (M\&O) algorithm, that
is a ``measure and react'' algorithm which maximizes expected
downloads at each decision point. In the \musiclab{}, the only degree
of freedom is the playlist presented to an incoming participant. Each
ranking policy for the playlist produces a ``measure and react''
algorithm. The original \musiclab{} experiments presented one such
policy: To rank the songs by number of downloads. Our ``measure and
optimize'' algorithm in contrast uses a performance ranking that
maximizes the expected downloads globally for the incoming
participant, exploiting the quality and appeal of each song, the
position visibilities in the playlist, and social influence. The
performance ranking is based on two insights.

\begin{enumerate}

\item It uses the generative \musiclab{} model
  that describes the probability of sampling a song given its quality,
  appeal, position in the playlist, and the downloads of all songs at
  some point in time.

\item It solves the resulting global optimization problem in strongly
  polynomial time by a reduction to the Linear Fractional Assignment
  Problem, meaning that the algorithm scales to large instances and
  avoid exploring the $n!$ possible rankings.
\end{enumerate}

\noindent
In addition, since the quality of the songs is not known a priori, our
\mo{} algorithm learns them over time, using an often overlooked
insight by Salganik et al. \cite{salganik2006experimental}: The fact
that the popularity of a song in the independent condition is a
natural measure of its quality, capturing both its intrinsic ``value''
and the preferences of the participants.

The \mo{} algorithm was evaluated on a simulated version of the
\musiclab{} based on the model proposed in
\cite{krumme2012quantifying}; It was compared to the download ranking
policy used in the original \musiclab{} experiments
\cite{salganik2006experimental} and a random ranking policy.  These
ranking policies were evaluated both under the social influence and
independent conditions.

Regarding efficiency, the computational results
contain two important highlights:
\begin{enumerate}
\item The \mo{} algorithm produces a significant increase in expected
  downloads compared to other tested policies.

\item The \mo{} algorithm leverages both position bias and social
  influence to maximize downloads, providing substantial gains
  compared to ranking policies under the independent condition.
\end{enumerate}

\noindent
These performance results are robust, not an artifact of distributions
and parameters used in experimental setting. Indeed, our theoretical
results show that a performance-ranking policy always benefits from
social influence and outperforms all policies not using social
influence in expectation.  The performance results also shed light on
recent results that demonstrated the benefits of position bias in the
absence of social influence \cite{Lerman2014}. Our computational
results show that a performance-ranking policy can leverage both
social influence and position visibility and that the improvements in
expected downloads are cumulative.

Regarding unpredictability, our results are also illuminating. They
highlight that the performance-ranking policy often identifies
``blockbusters'', which are the primary sources of the increased
market efficiency. This provides a sharp contrast to the download
ranking, which cannot identify these blockbusters reliably. Moreover,
even in that worst-case scenario, our computational results show that
the performance ranking brings significant benefits over the download
ranking.

Finally, it is important to emphasize that the performance-ranking
policy is general-purpose and can be applied to a variety of cultural
markets, including books and movies, where product appeal, position
visibility, and social influence can be optimized jointly.

\section{Results} 

We first describe the \musiclab{} model proposed in
\cite{krumme2012quantifying}, the experimental setting, and a
presentation of the various ranking policies. We then show how to
estimate the qualities of the songs. We then report the core
computational results about the performance and unpredictability of
the various ranking policies. We finally prove that social influence
always helps the performance ranking in maximizing the expected
downloads. As a result, the performance ranking under social influence
outperforms all ranking policies under the independent condition.

\subsection{The  \musiclab{} Model}
\label{section-predictive-model}

The descriptive model for the \musiclab{} \cite{krumme2012quantifying}
is based on the data collected during the actual experiments and is
accurate enough to reproduce the conclusions in
\cite{salganik2006experimental} through simulation. The model is
defined in terms of a market composed of $n$ songs. Each song $i \in
\{1,\ldots,n\}$ is characterized by two values:
\begin{enumerate}
\item Its {\em appeal} $A_i$ which represents the inherent preference
  of listening to song $i$ based only on its name and its band;

\item Its {\em quality} $q_i$ which represents the conditional
  probability of downloading song $i$ given that it was sampled.
\end{enumerate}
The \musiclab{} experiments present each participant with a playlist
$\pi$, i.e., a permutation of $\{1,\ldots,n\}$. Each position $p$ in
the playlist is characterized by its {\em visibility} $v_p$ which is
the inherent probability of sampling a song in position $p$. Since the
playlist $\pi$ is a bijection from the positions to the songs, its
inverse is well-defined and is called a ranking. Rankings are denoted
by the letter $\sigma$ in the following, $\pi_i$ denotes the song in
position $i$ of the playlist $\pi$, and $\sigma_i$ denotes the
position of song $i$ in the ranking $\sigma$. Hence $v_{\sigma_i}$
denotes the visibility of the position of song $i$. The model
specifies the probability of listening to song $i$ at time $k$ given a
playlist $\sigma$ as
\[
p_{i,k}(\sigma) =  \frac{v_{\sigma_i}(\alpha A_i+D_{i,k})}{\sum_{j=1}^n v_{\sigma_j}(\alpha A_j+D_{j,k})},
\]
where $D_{i,k}$ is the number of downloads of song $i$ at time $k$ and
$\alpha > 0$ is a scaling factor which is the same for all
songs. Observe that the probability of sampling a song depends on its
position in the playlist, its appeal, and the number of downloads at
time $k$. As an experiment proceeds, the number of downloads dominates
the appeal of the song at a rate dictated by $\alpha$. 

When simulating the independent condition, the download counts
$D_{i,k}$ in the probability $p_{i,k}(\sigma)$ becomes zero and hence
the probability of listening to song $i$ only depends on the appeals
and visibilities of the songs. We obtain a probability
\[
p_{i}(\sigma) =  \frac{v_{\sigma_i} A_i}{\sum_{j=1}^n v_{\sigma_j}A_j}
\]

\subsection{The Experimental Setting}
\label{section-experimental-setting}

The experimental setting in this paper uses an agent-based simulation
to emulate the \musiclab{}. Each simulation consists of $N$
iterations and, at each iteration $k$,

\begin{enumerate}
\item the simulator randomly selects a song $i$ according to the
  probability distribution specified by the $p_{j.k}$'s.

\item the simulator randomly determines, with probability $q_i$,
  whether selected song $i$ is downloaded; If the song is downloaded,
  then the simulator increases the download counts of song $i$, i.e.,
  $D_{i,k+1} = D_{i,k} + 1$.
\end{enumerate}

\noindent
The ``measure and react'' algorithms studied in this paper monitors
the above process (e.g., to estimate song quality). Every $r$
iterations, the algorithms compute a new playlist $\sigma$ using one
of the ranking policies described below.  For instance, in the social
influence condition of the original \musiclab{} experiments, the
policy ranks the songs in decreasing order of download counts. The
parameter $r \geq 1$ is called the refresh rate.

\begin{figure}[t]
\begin{centering}
\includegraphics[width=0.4\linewidth]{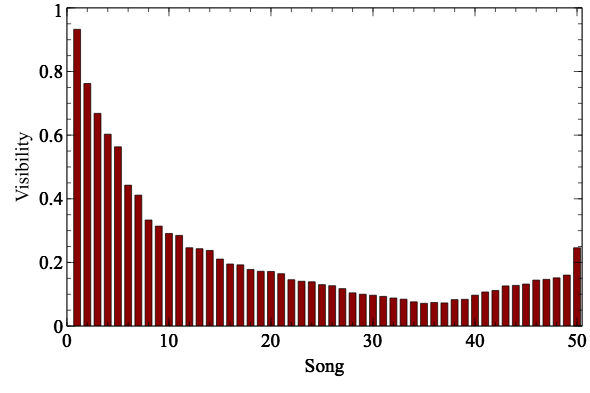}
\end{centering}
\vspace{-0.5cm}
\centering{}
\caption{The visibility $v_p$ (y-axis) of position $p$ in the
  playlist (x-axis) where $p = 1$ is the top position and $p = 50$ is
  the bottom position of the playlist which is displayed in a single column.}
\label{fig:visibility}
\end{figure}

\begin{figure}[t]
\begin{centering}
\includegraphics[width=0.45\linewidth]{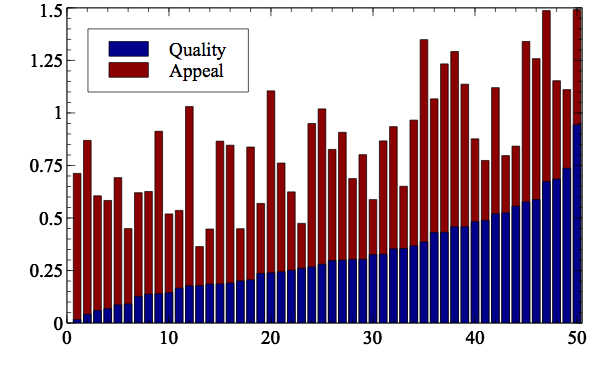}
\includegraphics[width=0.45\linewidth]{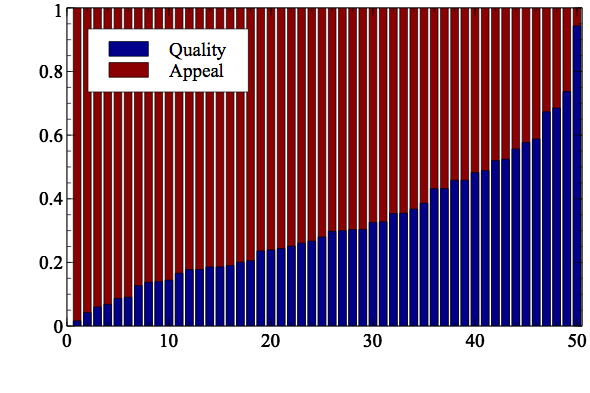}
\par\end{centering}
\vspace{-0.5cm}
\centering{}
\caption{ The Quality $q_i$ (blue) and Appeal $A_i$ (red) of Song $i$
  in the two settings. In the first setting, the quality and the
  appeal for the songs were chosen independently according to a
  Gaussian distribution normalized to fit between $0$ and $1$. The
  second setting explores an extreme case where the quality is
  negatively correlated with the appeal.}
\label{fig:qualityAppeal}
\end{figure}

The experimental setting tries to be close to the \musiclab{}
experiments and considers 50 songs and simulations with 20,000
steps. The songs are displayed in a single column. The analysis in
\cite{krumme2012quantifying} indicated that participants are more
likely to sample songs higher in the playlist. More precisely, the
visibility decreases with the playlist position, except for a slight
increase at the bottom positions. Figure \ref{fig:visibility} depicts
the visibility profile based on these guidelines, which is used in all
computational experiments in this paper.

The paper also uses two settings for the quality and appeal of each
song, which are depicted in Figure \ref{fig:qualityAppeal}:
\begin{enumerate}
\item In the first setting, the quality and the appeal for the songs were chosen
independently according to a Gaussian distribution normalized to fit
between $0$ and $1$. 

\item The second setting explores an extreme case where
the quality is negatively correlated with the appeal.
\end{enumerate}
The experimental results were obtained by running $W=400$ simulations.

\subsection{The Ranking  Policies}
\label{section-rankings}

This section reviews a number of ranking policies used to influence
the participants. Each policy updates the order of the songs in the
playlist and thus defines a different ``measure and react'' algorithm that uses
the policy to present a novel playlist after each refresh.

The first policy is called the {\em download ranking}
\cite{salganik2006experimental}: At iteration $k$, it simply orders
the songs by the number of downloads $D_{i,k}$ and assigns the
positions $1..n$ in that order to obtain the $\sigma_k^d$.  Note that
downloads do not necessarily reflect the inherent quality of the
songs, since they depend on how many times the songs were listened to.

The second policy is the {\em performance ranking} which maximizes the
expected number of downloads at each iteration, exploiting all the
available information globally, i.e., the appeal, the visibility, the
downloads, and the quality of the songs.

\begin{defn}\label{def:opt} The performance ranking $\sigma_k^*$ at iteration
  $k$ is defined as
\[
\sigma_k^* = \arg\max_{\sigma\in S_n} \sum_{i=1}^n p_{i,k}(\sigma)\cdot q_i.
\]
\end{defn}

\noindent
The performance ranking uses the probability $p_{i,k}(\sigma)$ of
sampling song $i$ at iteration $k$ given ranking $\sigma$, as well as
the quality $q_i$ of song $i$. Obviously, in practical situations, the
song qualities are unknown. However, as mentioned in the introduction,
Salganik et al. \cite{salganik2006experimental} argue that the
popularity of a song (at least in the independent condition) can be
used as a measure of its quality. Hence, the quality $q_{i}$ in the
above expression can be replaced by its approximation $\hat{q}_{i,k}$
at iteration $k$ as discussed later in the paper.

It is not obvious a priori how to compute the performance ranking: In
the worst case, there are $n!$ rankings to explore. Fortunately, the
performance ranking can be computed efficiently (in strongly
polynomial time) through a reduction to the \emph{Linear Fractional
  Assignment Problem}.

\begin{theorem}
The performance ranking can be computed in strongly polynomial time.
\end{theorem}

\begin{proof}
  We show that the performance ranking problem can be polynomially
  reduced to the Linear Fractional Assignment Problem (LFAP). The
  LFAP, which we state below, is known to be solvable in polynomial
  time in the strong sense
  \cite{shigeno1995algorithm,kabadi2008strongly}.

  A LFAP input is a bipartite graph $G=(V_1,V_2,E)$ where
  $|V_1|=|V_2|=n$, and a cost $c_{ij}$ and a weight $d_{ij}>0$ for
  each edge $(i,j)\in E$. The LFAP amounts to finding a perfect
  matching $M$ minimizing
\[
\frac{\sum_{(i,j)\in M}c_{ij}}{\sum_{(i,j)\in M}d_{ij}}.
\]
To find the performance ranking, it suffices to solve the LFAP problem
defined over a bipartite graph where $V_1$ represents the songs, $V_2$
represents the positions, $c_{i,j}= -v_j(\alpha A_i + D_i)q_i$ and
$d_{ij} = v_j(\alpha A_i + D_i)$. This LFAP instance determines a
perfect matching $M$ such that $\frac{- \sum_{(i,j)\in M} v_j(\alpha
  A_i + D_i)q_i}{\sum_{(i,j)\in M} v_j(\alpha A_i + D_i)}$ is
minimized. The result follows since this is a polynomial reduction.
\end{proof}

\noindent
In the following, we use {\sc D-rank} and {\sc P-rank} to denote the
``measure and react'' algorithms using the download and performance
rankings respectively. We also annotate the algorithms with either
{\sc SI} or {\sc IN} to denote whether they are used under the social
influence or the independent condition. For instance, {\sc P-rank(SI)}
denotes the algorithm with the performance ranking under the social
influence condition, while {\sc P-rank(IN)} denotes the algorithm with
the performance ranking under the independent condition. We also use
{\sc rand-rank} to denote the algorithm that simply presents a random
playlist at each iteration. Our proposed \mo{} algorithm is {\sc
  P-rank(SI)}.

Under the independent condition, the optimization problem is the same
at each iteration as the probability $p_{i}(\sigma)$ does not change
over time. Since the performance ranking maximizes the expected
downloads at each iteration, it dominates all other ``measure and
react'' algorithms under the independent condition.

\begin{proposition}[Optimality of Performance Ranking] Given specified
  songs appeal and quality, {\sc P-rank(IN)} is the optimal ``measure
  and react'' algorithm under the independent condition.
\end{proposition}

\subsection{Recovering the Songs Quality}
\label{section-recovering-quality}

This section describes how to recover songs quality in the
\musiclab{}. It shows that ``measure and react'' algorithms quickly
recover the real quality of the songs.

As mentioned several times already, Salganik et
al. \cite{salganik2006experimental} stated that the popularity of a
song in the independent condition is a natural measure of its quality
and captures both its intrinsic ``value'' and the preferences of the
participants. To approximate the quality of a song, it suffices to
sample the participants. This can be simulated by using a Bernoulli
sampling based on real quality of the songs. The predicted quality
$\hat{q}_i$ of song $i$ is obtained by running $m$ independent
Bernoulli trials with probability $q_i$ of success, i.e.,
\[
\hat{q}_i=\frac{k}{m},
\]
where $k$ is the number of successes over the $m$ trials.  The law of
large numbers states that the prediction accuracy increases with the
sampling size. For a large enough sampling size, $\hat{q}_i$ has a
mean of $q_i$ and a variance of $q_i(1-q_i)$. This variance has the
desirable property that the quality of a song with a more 'extreme'
quality (i.e., a good or a bad song) is recovered faster than those
with average quality.  In addition, ``measure and react'' algorithms
merge additional information about downloads into the prediction as
the experiment proceeds: At step $k$, the approximate quality of song
$i$ is given by
\[
\hat{q}_{i,k} = \frac{\hat{q}_{i,0} \cdot m + D_{i,k}}{m + S_{i,k}},
\]
where $m$ is the initial sample size and $S_{i,k}$ the number of
samplings of song $i$ up to step $k$.

\begin{figure}[ht]
\begin{centering}
\includegraphics[width=0.8\linewidth]{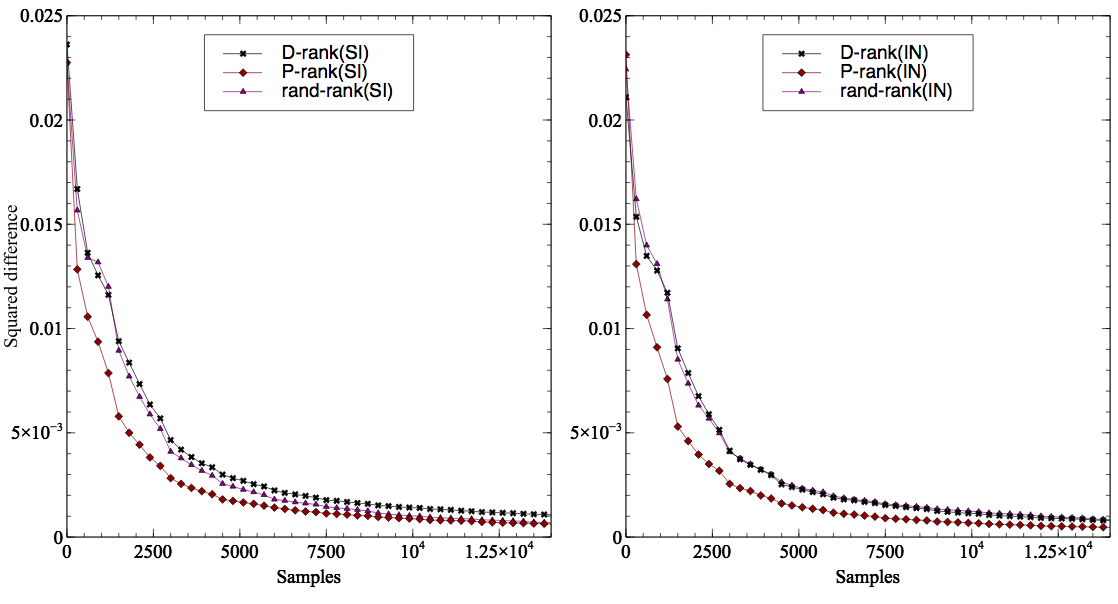}
\par\end{centering}
\centering{}
\caption{Average Squared Difference of Inferred Quality over Time for
  Different Rankings for the top 10 quality songs. The figure reports the
  average squared difference $ \sum_{i=1}^{n} \frac{(\hat{q}_{i,k} -
    q_{i})^2}{n} $ between the song quality and their predictions for
  the various ranking policies under the social influence and the
  independent conditions.}
\label{fig:accuracy_incom2}
\end{figure}

Figure \ref{fig:accuracy_incom2} presents experimental results about
the accuracy of the quality approximation for different rankings,
assuming an initial sampling set of size 10. More precisely, the figure
reports the average squared difference
\[
\sum_{i=1}^{n} \frac{(\hat{q}_{i,k} - q_{i})^2}{n}
\]
between the song quality and their predictions for the various ranking
policies under the social influence and the independent conditions. In
all cases, the results indicate that the songs quality is recovered
quickly and accurately. Observe that the performance ranking recovers
the song quality faster than the other rankings, both in the social
influence and in the independent conditions.

\subsection{Expected Performance}
\label{section-expected-performance}

This section studies the expected performance of the various rankings
under the social influence and the independent conditions. The
performance ranking uses the quality approximation presented in the
last section (instead of the ``true'' quality). Figures
\ref{fig:down_incom} and \ref{fig:down_incom2} depict the average
number of downloads for the various rankings given the quality and
appeal data depicted in Figure \ref{fig:qualityAppeal}.

\begin{figure}[t]
\begin{centering}
\includegraphics[width= 0.8 \linewidth]{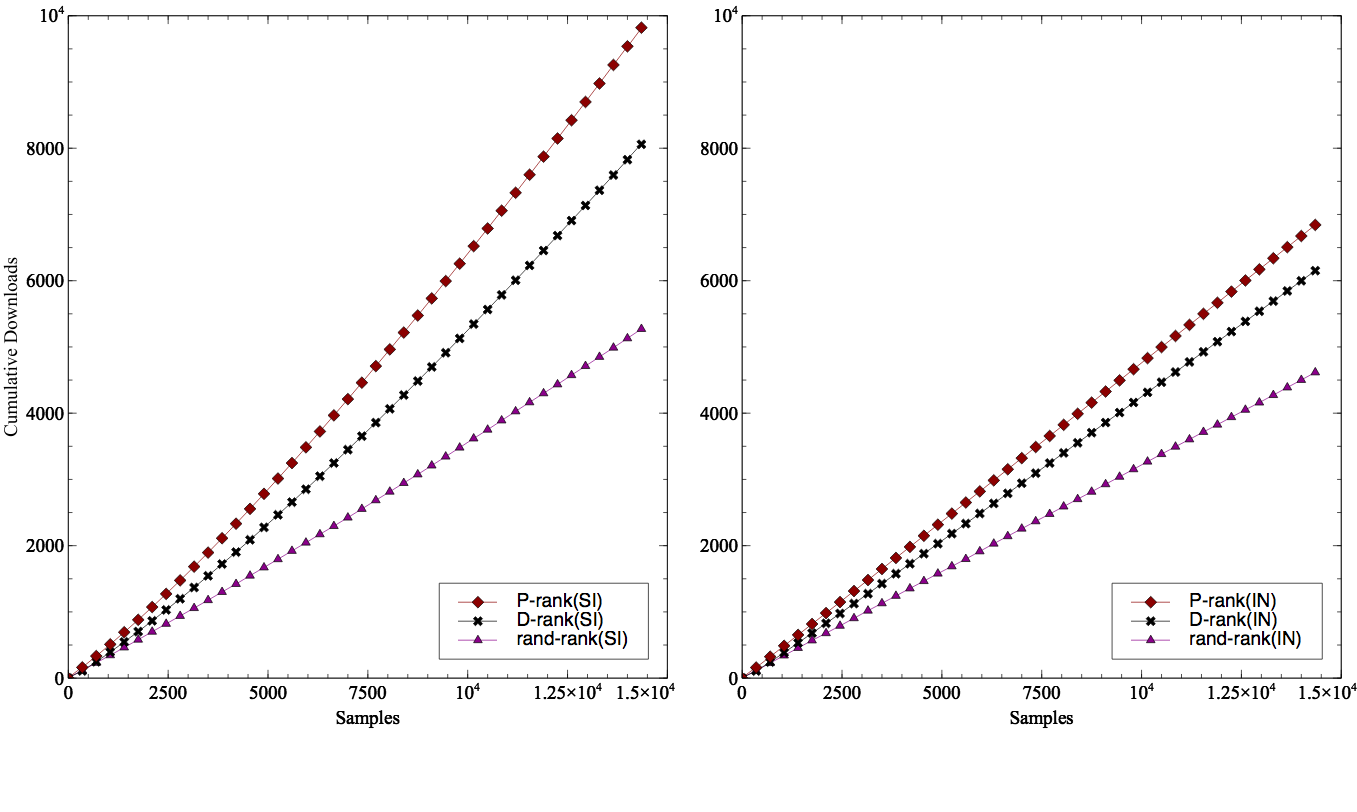}
\par\end{centering}
\vspace{-0.5cm}
\centering{}
\caption{The Number of Downloads over Time for the Various Rankings
  (Independent Gaussian Setting). The x-axis represents the number of song samplings
  and the y-axis represents the average number of downloads over all
  experiments. The left figure depicts the results for the social
  influence condition and the right figure for the independent
  condition.}
\label{fig:down_incom}
\end{figure}

\begin{figure}[t]
\begin{centering}
\includegraphics[width=0.8 \linewidth]{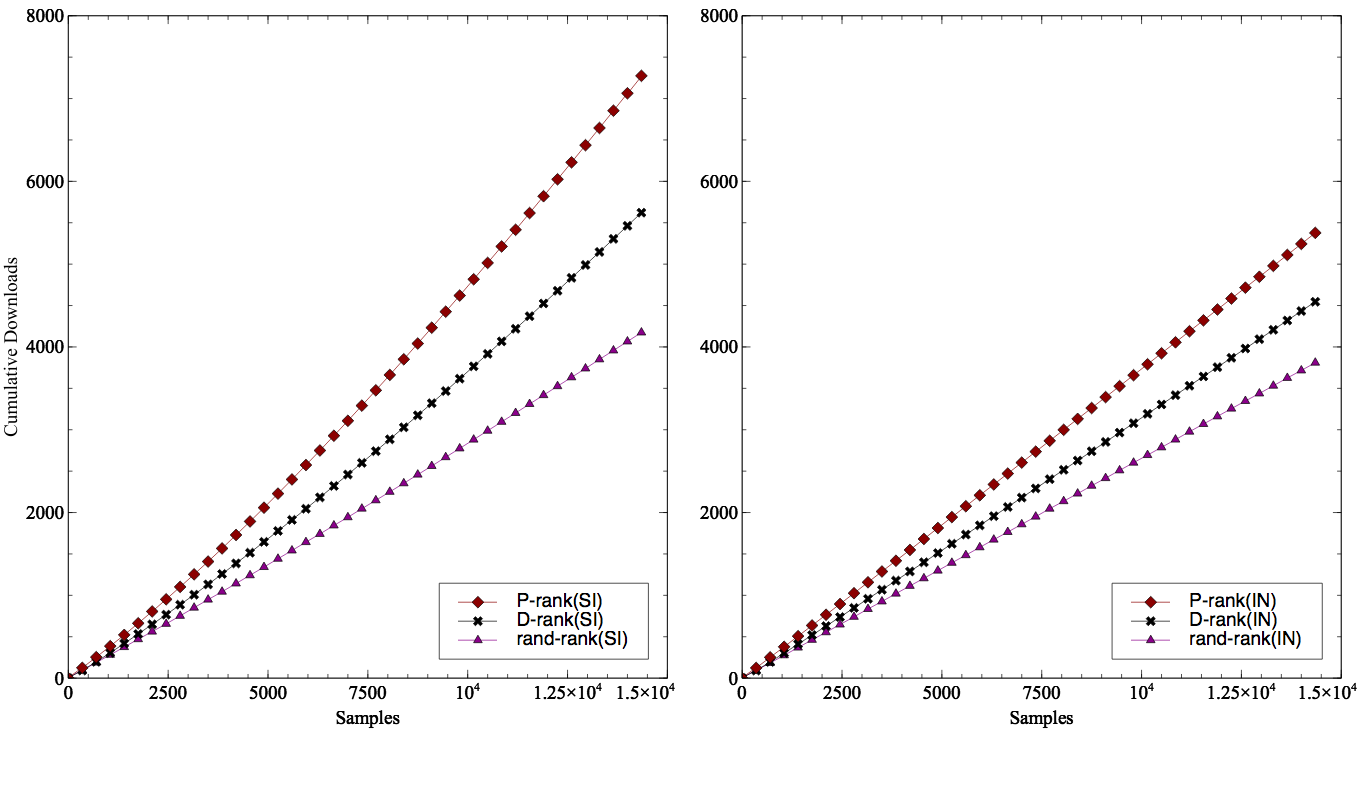}
\par\end{centering}
\centering{}
\vspace{-0.5cm}
\caption{The Number of Downloads over Time for the Various Rankings
  (Negative Correlation Setting). The x-axis represents the number of
  song samplings and the y-axis represents the average number of
  downloads over all experiments. The left figure depicts the results
  for the social influence condition and the right figure for the
  independent condition.}
\label{fig:down_incom2}
\end{figure}

The experimental results highlight three significant findings:
\begin{enumerate}

\item The {\em performance ranking} provides substantial gains in
  expected downloads compared to the {\em download ranking} and the
  {\em random ranking} policies, both in the social influence and
  the independent conditions.

\item The benefits of the performance ranking are particularly
  striking under the social influence condition. Social influence
  creates a Matthew effect \cite{MatthewBook} for the number of
  expected downloads.

\item The gain of the performance ranking is also more significant on
  the worst-case scenario, where the song appeal is negatively
  correlated with the song quality. The expected number of downloads
  is lower in this worst-case scenario, but the shapes of the curves
  are in fact remarkably similar in the social influence and
  independent conditions.
\end{enumerate}

\noindent
In summary, the performance ranking produces significant improvements
in expected downloads over the download ranking, which also dominate
the random ranking. The benefits come both from position bias and
social influence which provide cumulative improvements.

\subsection{Market Unpredictability}

Figure \ref{fig:unpred1} depicts the unpredictability results using
the measure proposed in \cite{salganik2006experimental}. The
unpredictability $u_i$ of song $i$ is defined as the average
difference in market share for that song between all pairs of
realizations, i.e.,
\[
u_i = \sum^W_{w=1}\sum^W_{w'=w+1}|m_{i,w}-m_{i,w'}|/\binom{W}{2},
\]
where $m_{i,w}$ is the market share of song $i$ in world $w$, i.e.,
\[
m_{i,j} = D_{i,N}^j /\sum^n_{k=1} D_{i,N}^j
\]
where $D_{i,t}^w$ is the number of downloads of song $i$ at step $t$
in world $w$. The overall unpredictability measure is the average of
this measure over all $n$ songs, i.e.,
\[
U= \sum^n_{j=1}u_i/n.
\]

\begin{figure}[t]
\begin{centering}
\includegraphics[width=0.47\linewidth]{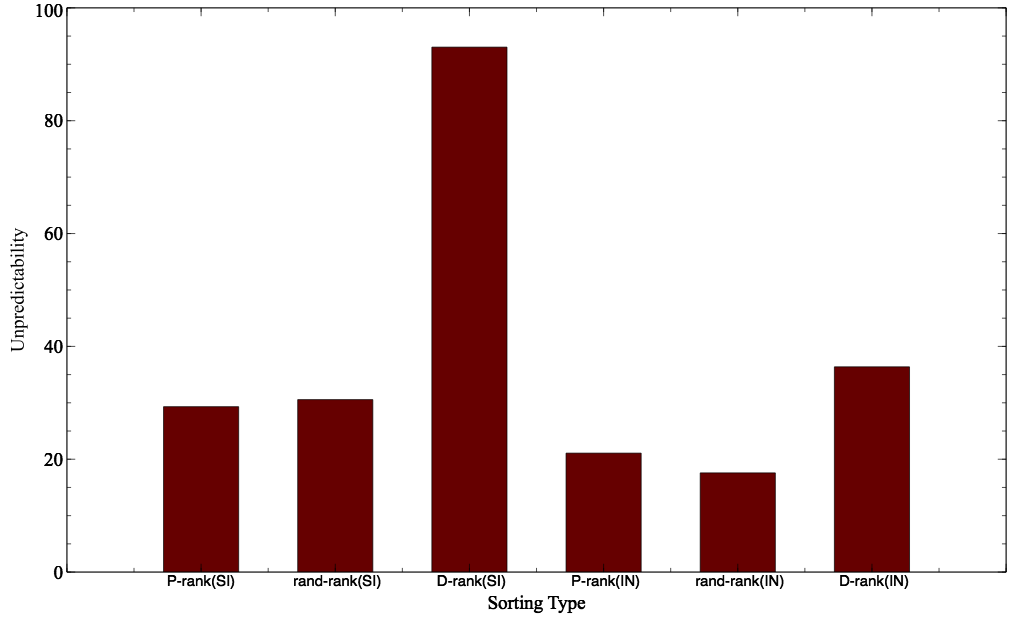}
\includegraphics[width=0.47\linewidth]{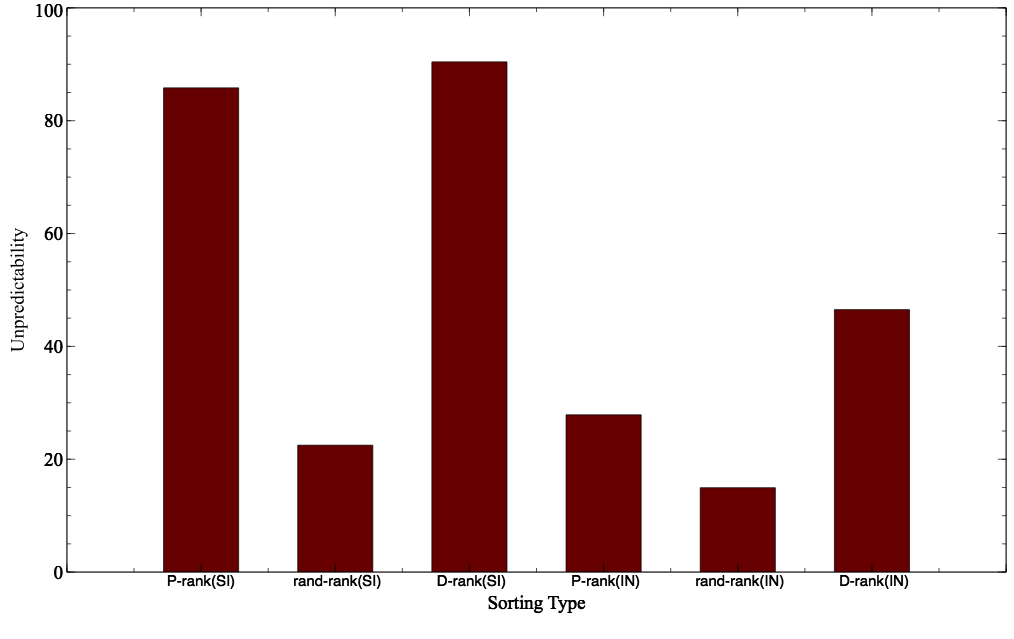}
\par\end{centering}
\centering{}
\vspace{-0.5cm}
\caption{The Unpredictability of the Rankings for the Unpredictability
  Measure of \cite{salganik2006experimental}. The measure of
  unpredictability $u_i$ for song $i$ is defined as the average
  difference in market share for that song between all pairs of
  realizations, i.e., $ u_i =
  \sum^W_{j=1}\sum^W_{k=j+1}|m_{i,j}-m_{i,k}|/\binom{W}{2}$, where
  $m_{i,j}$ is the market share of song $i$ in world $j$.  The overall
  unpredictability measure is the average of this measure over all $n$
  songs, i.e., $ U= \sum^n_{j=1}u_i/n.  $. The left figure depicts the
  results for the first setting, while the right figure depicts the
  results for the second setting.}
\label{fig:unpred1}
\end{figure}

\noindent
The figure highlights two interesting results.
\begin{enumerate}
\item Exploiting social influence does not necessarily create an
  unpredictable market.  In the first setting, the performance ranking
  introduces only negligible unpredictability in the market. This
  obviously contrasts with the download ranking, which introduces
  significant unpredictability.

\item An unpredictable market is not necessarily an inefficient
  market. In the second setting, the performance ranking exploits
  social influence to provide significant improvements in downloads
  although, as a side-effect, it creates more
  unpredictability. Subsequent results will explain the source of this
  unpredictability.
\end{enumerate}

\subsection{Downloads Versus Quality}

\begin{figure}[t]
\begin{centering}
\includegraphics[width=0.7 \linewidth]{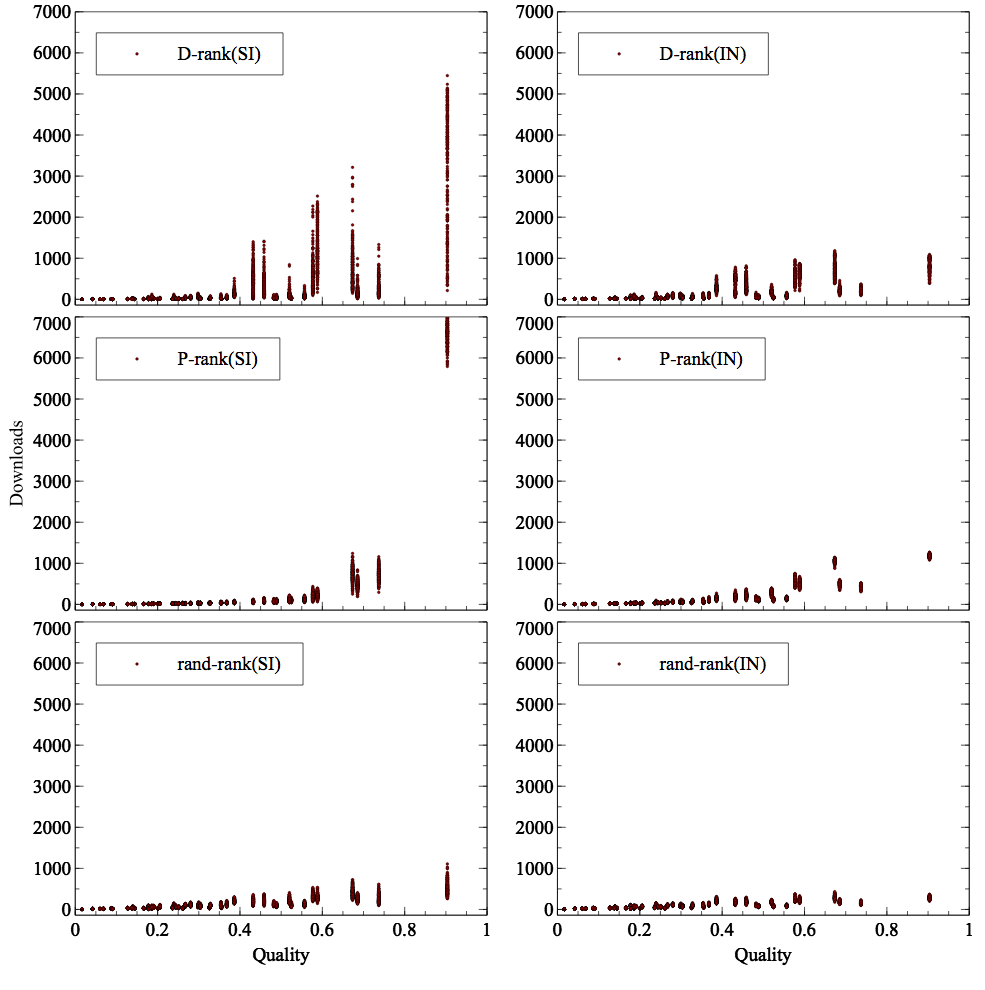}
\par\end{centering}
\centering{}\caption{The Distribution of Downloads Versus the True
  Song Qualities (First Setting). The songs on the x-axis are ranked
  by increasing quality from left to right. Each dot is the number of
  downloads of a song in an experiment.}
\label{fig:distribution}
\label{fig:Dist1}
\end{figure}

\begin{figure}[t]
\begin{centering}
\includegraphics[width=0.7 \linewidth]{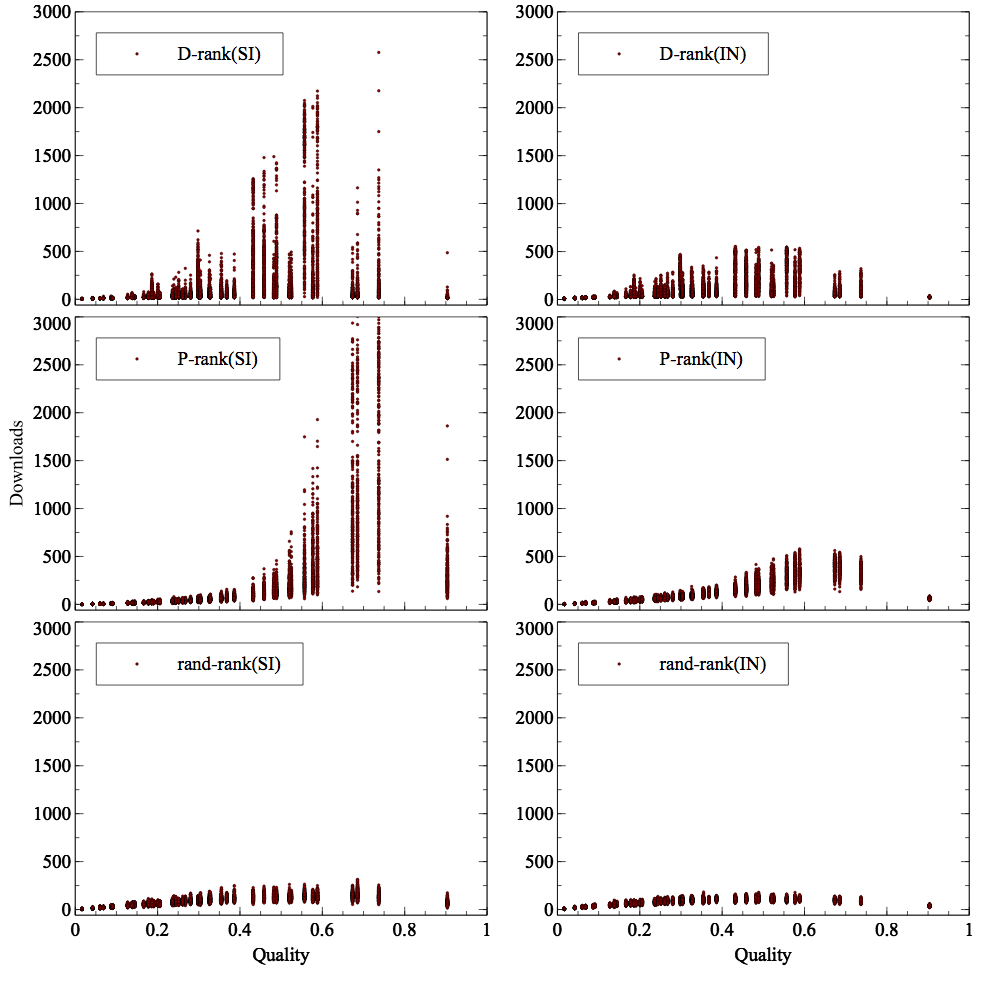}
\par\end{centering}
\centering{}\caption{The Distribution of Downloads Versus the True
  Song Qualities (Second Setting). The songs on the x-axis are ranked
  by increasing quality from left to right. Each dot is the number of
  downloads of a song in an experiment.}
\label{fig:distribution2}
\label{fig:Dist2}
\end{figure}

Figures \ref{fig:Dist1} and \ref{fig:Dist2} report the correlation
between number of downloads of a song (y-axis) and its true quality
(x-axis) for each of the rankings. More precisely, the figures show
the distributions of the downloads of every song over the simulations.
The songs are ranked in increasing quality on the x-axis and there are
400 dots for each song, representing the number of downloads in each
of the experiments. As a result, the figures give a clear visual
representation of what happens in the \musiclab{} for the various
rankings, including how the downloads are distributed and the sources
of unpredictability.

The results are particularly interesting. In the first setting
depicted in Figure \ref{fig:Dist1}, {\sc P-rank(SI)} clearly isolates
the best song, which has an order of magnitude more downloads than
other songs. When comparing with {\sc P-rank(SI)} and {\sc
  P-rank(IN)}, we observe that the additional market efficiency
produced by {\sc P-rank(SI)} primarily comes from the downloads of
this best song. In contrast, the download ranking has a lot more
unpredictability and has significant download counts for many other
songs, even for some with considerably less quality. The variance of
downloads for the best song is substantial, indicating that there are
simulations in which this best song has very few downloads. This never
happens with the performance ranking.

The results in Figure \ref{fig:Dist2} highlight the source of
unpredictability in the performance ranking. These results are for the
second setting where the appeal is negatively correlated with the
quality. The results indicate that the performance ranking is not
capable of isolating the best song in this setting. There are
significant downloads of the next high-quality songs however. This
also contrasts with the download ranking where lower-quality songs can
have substantial downloads. It is interesting to point out that the
unpredictability of the performance ranking in this setting comes from
songs which are not easy to distinguish: These songs have very similar
values for quality and appeal.

These results shed interesting light on social influence. First, the
download ranking clearly shows that social influence can transform
average songs into ``hits''. However, this behavior is entirely
eliminated by the performance ranking in the first setting: The
performance ranking clearly leverages social influence to isolate the
``blockbuster''. The second setting highlights a case where even the
performance ranking cannot recover the best song. This setting is of
course a worst-case scenario: A terrific product with a poor
appeal. However, the performance ranking still leverages social
influence to promote high-quality songs: Social influence under the
performance ranking is much less likely to turn an average song into a
``hit'' compared to the download ranking.

\subsection{The Benefits of Social Influence}

We now prove that social influence always increases the number of
expected downloads of the performance ranking. We show that this is
true regardless of the number of songs, their appeal and quality, and
the visibility values. The proof assumes that the quality of the songs
has been recovered, i.e., we use $q_i$ instead of $\hat{q}_{i,k}$.
Also, for simplifying the notations, we use $a_{i}=\alpha
A_{i}+D_{i,t}$ and assume without loss of generality that
$v_{1}\geq\cdots\geq v_{n}>0$.

We are interested in showing that the expected number of downloads
increases over time for the performance ranking under the social
influence condition. In state $t$, the probability that song $i$ is
downloaded for ranking $\sigma$ is
\[
\frac{v_{\sigma_i}a_{i}q_{i}}{\sum
  v_{\sigma_i}a_{i}}.
\]
Hence, the expected number of downloads at time $t$ for ranking
$\sigma$ is defined as
\[
\mathbb{E}[D_{t}^{\sigma}]=\frac{\sum_{i}v_{\sigma_i}a_{i}q_{i}}{\sum_{i}v_{\sigma_i}a_{i}}.
\]
Under the social influence condition, the number of expected downloads
over time can be considered as a Markov chain where state $t+1$ only
depends on state $t$. Therefore the expected number of downloads at
time $t+1$ conditional to time $t$ if ranking $\sigma$ and $\sigma'$
are used at time $t$ and $t+1$ respectively is
\[
\mathbb{E}[D_{t+1}^{\sigma,\sigma'}]=\sum_{j}\left(\frac{v_{\sigma_j}a_{j}q_{j}}{\sum
    v_{\sigma_i}a_{i}}\cdot\frac{\sum_{i\not=j}v_{\sigma'_i}a_{i}q_{i}+v_{\sigma'_j}(a_{j}+1)q_{j}}{\sum_{i\not=j}v_{\sigma'_i}a_{i}+v_{\sigma'_j}(a_{j}+1)}\right)+\left(1-\mathbb{E}[D_{t}^{\sigma}]\right)\cdot
\mathbb{E}[D_{t}^{\sigma}].
\]
Consider the performance ranking now. At step $t$, the performance ranking finds
a permutation $\sigma^{*}\in S_{n}$ such that
\begin{equation}
\mathbb{E}[D_{t}^{\sigma^*}] = \frac{\sum_{i}v_{\sigma^*_i}a_{i}q_{i}}{\sum_{i}v_{\sigma^*_i}a_{i}} =
\max_{\sigma} \frac{\sum_{i}v_{\sigma_i}a_{i}q_{i}}{\sum_{i}v_{\sigma_i}a_{i}}. \label{eq:optp}
\end{equation}
By optimality of the performance ranking at each step, we have that
\begin{equation}
\mathbb{E}[D_{t+1}^{\sigma^*,\sigma^{**}}] \geq \mathbb{E}[D_{t+1}^{\sigma^*,\sigma^*}] \label{eq:same}
\end{equation}
where $\sigma^{**}$ is the performance ranking at step $t+1$. Hence, to show that
social influence is always beneficial to the performance ranking, it suffices to
show that 
\[
\mathbb{E}[D_{t+1}^{\sigma^*,\sigma^*}] \geq \mathbb{E}[D_{t}^{\sigma^*}].
\]
The proof is simpler when reasoning about the playlist obtained by the
performance ranking. At step $t$, the performance ranking finds a
playlist $\pi^{*}\in S_{n}$ such that
\begin{equation}
\frac{\sum_{p}v_{p}a_{\pi^*_p}q_{\pi^*_p}}{\sum_{p}v_{p}a_{\pi^*_p}} = \mathbb{E}[D_{t}^{\sigma^*}]  \label{eq:optplaylist}
\end{equation}

\noindent
By (\ref{eq:optplaylist}), $\pi^{*}$ satisfies
\begin{equation}
\sum_{p}v_{p}a_{\pi_{p}^{*}}q_{\pi_{p}^{*}}-\mathbb{E}[D_{t}^{\sigma^*}] \ \sum_{p}v_{p}a_{\pi_{p}^{*}}=\sum_{p}v_{p}a_{\pi_{p}^{*}}\left(q_{\pi_{p}^{*}}-\mathbb{E}[D_{t}^{\sigma^*}]\right) = 0. \label{eq:opt}
\end{equation}

\noindent
Consider the function
\[
f(\lambda)=\max_{\pi}\sum_{p}v_{p}a_{\pi_{p}^{*}}\left(q_{\pi_{p}^{*}}-\lambda\right).
\]
Function $f$ is concave and strictly decreasing and hence its unique
zero is the optimal value $\mathbb{E}[D_{t}^{\sigma^*}]$.
Furthermore, the optimality condition for $\pi$ in function $f$ can be
characterized by a \textquotedbl{} rearrangement
inequality\textquotedbl{} which states that the maximum is reached if
and only if
\begin{equation}
a_{\pi_{1}^{*}}\left(q_{\pi_{1}^{*}}-\lambda\right)\geq\cdots\geq
a_{\pi_{n}^{*}}\left(q_{\pi_{n}^{*}}-\lambda\right)
\quad\forall\lambda\in\mathcal{R}. \label{eq:rearrangement}
\end{equation}

\noindent
We are now in position to state our key theoretical result.
\begin{thm}
\label{thm:increasing}
The expected downloads are nondecreasing over time when the performance ranking
is used to select the playlist.
\end{thm}
\begin{proof}
By Equation \ref{eq:same}, it suffices to show that 
\[
\mathbb{E}[D_{t+1}^{\sigma^*,\sigma^*}] \geq \mathbb{E}[D_{t}^{\sigma^*}].
\]
Let $a=a_{1},\cdots,a_{n}$ represent the current state at time $t$.
For ease of notation, we rename the songs so that $\sigma^*_i = i$. We
also drop the ranking superscript $\sigma^*$ from the download
notation. Denote the optimal expected number of downloads at time $t$
as
\[
\mathbb{E}[D_{t}]=\frac{\sum_{i}v_{i}a_{i}q_{i}}{\sum_{i}v_{i}a_{i}}=\lambda^{*}.
\]
Then, the expected number of downloads in time $t+1$ conditional to time $t$
is
\[
\mathbb{E}[D_{t+1}]=\sum_{j}\left(\frac{v_{j}a_{j}q_{j}}{\sum v_{i}a_{i}}\cdot\frac{\sum_{i\not=j}v_{i}a_{i}q_{i}+v_{j}(a_{j}+1)q_{j}}{\sum_{i\not=j}v_{i}a_{i}+v_{j}(a_{j}+1)}\right)+\left(1-\frac{\sum_{i}v_{i}a_{i}q_{i}}{\sum_{i}v_{i}a_{i}}\right)\cdot\frac{\sum_{i}v_{i}a_{i}q_{i}}{\sum_{i}v_{i}a_{i}}
\]
\[
=\sum_{j}\left(\frac{v_{j}a_{j}q_{j}}{\sum v_{i}a_{i}}\cdot\frac{\sum_{i}v_{i}a_{i}q_{i}+v_{j}q_{j}}{\sum_{i}v_{i}a_{i}+v_{j}}\right)+\left(1-\frac{\sum_{j}v_{j}a_{j}q_{j}}{\sum_{i}v_{i}a_{i}}\right)\cdot\lambda^{*}.
\]
Proving
\begin{equation}
\mathbb{E}[D_{t+1}] \geq \mathbb{E}[D_{t}]\label{eq:1}
\end{equation}
amounts to showing that 
\[
\sum_{j}\left(\frac{v_{j}a_{j}q_{j}}{\sum v_{i}a_{i}}\cdot\frac{\sum_{i}v_{i}a_{i}q_{i}+v_{j}q_{j}}{\sum_{i}v_{i}a_{i}+v_{j}}\right)+\left(1-\frac{\sum_{j}v_{j}a_{j}q_{j}}{\sum_{i}v_{i}a_{i}}\right)\cdot\lambda^{*} \geq \lambda^{*}.
\]
which reduces to proving
\[
\frac{1}{\sum_{i}v_{i}a_{i}}\sum_{j}\left[\frac{v_{j}^2a_{j}q_{j}}{\sum_{i}v_{i}a_{i}+v_{j}}\left(q_{j}-\lambda^{*}\right)\right] \geq 0
\]
or, equivalently, 
\begin{equation}\label{eq:condition1}
\sum_{j}\left[\frac{v_{j}^2a_{j}q_{j}}{\sum_{i}v_{i}a_{i}+v_{j}}\left(q_{j}-\lambda^{*}\right)\right] \geq 0.
\end{equation}

\noindent
By the rearrangement inequality (\ref{eq:rearrangement}), the performance ranking at step $t$ produces a
playlist $\pi$ satisfying
\[
a_{\pi_{1}^{*}}\left(q_{\pi_{1}^{*}}-\lambda^{*}\right)\geq\cdots\geq a_{\pi_{n}^{*}}\left(q_{\pi_{n}^{*}}-\lambda^{*}\right).
\]
Therefore, if $a_{i}\left(q_{i}-\lambda^{*}\right)\geq
a_{j}\left(q_{j}-\lambda^{*}\right)$, then $v_{i}\geq v_{j}$ for any
$i,j\in N$. Hence, the performance rank allocates all songs with a
negative term after the songs with a positive term. Define
$P_{1}=\left\{ i\in N|\left(q_{i}-\lambda^{*}\right)\geq0\right\}$ and
$P_{2}=\left\{ i\in N|\left(q_{i}-\lambda^{*}\right)<0\right\}$. It
follows that $v_{i}\geq v_{j}$ for all $i\in P_{1}$, $j\in P_{2}$.
Inequality (\ref{eq:condition1}) can be rewritten as follows:
\[
\sum_{j\in P_{1}}\left[\frac{v_{j}a_{j}q_{j}v_{j}}{\sum_{i}v_{i}a_{i}+v_{j}}\left(q_{j}-\lambda^{*}\right)\right]+\sum_{j\in P_{2}}\left[\frac{v_{j}a_{j}q_{j}v_{j}}{\sum_{i}v_{i}a_{i}+v_{j}}\left(q_{j}-\lambda^{*}\right)\right] \geq 0.
\]
By definition of $P_{1}$ and $P_{2}$,  the terms in the left summation are positive and
the terms in the right summation are negative.  Additionally, for any $c>0$,
\[
\frac{v_{i}}{c+v_{i}}\geq\frac{v_{j}}{c+v_{j}}\Leftrightarrow(c+v_{j})v_{i}\geq(c+v_{i})v_{j}\Leftrightarrow cv_{i}\geq cv_{j}\Leftrightarrow v_{i}\geq v_{j}
\]
Hence, since $v_{1}\geq v_{2}\geq\ldots\geq v_{n}$, 
\[
\frac{v_{1}}{\sum_{i}v_{i}a_{i}+v_{1}}\geq\frac{v_{2}}{\sum_{i}v_{i}a_{i}+v_{2}}\geq\ldots\geq\frac{v_{n}}{\sum_{i}v_{i}a_{i}+v_{n}}.
\]
Let $\underline{v}=\min_{i\in P_{1}}v_{i}$ the smallest visibility
assigned to an item in $P_{1}$. By the rearrangement inequality
(\ref{eq:rearrangement}), $\underline{v}\geq\max_{i\in
  P_{2}}v_{i}$. Similarly, let $\underline{q}=\min_{i\in
  P_{1}}q_{i}$. By definiton of $P_{1}$ and $P_{2}$,
$\underline{q}\geq\max_{i\in P_{2}}q_{i}$.  These observations leads
to the following inequality:
\[
\sum_{j\in P_{1}}\left[\frac{v_{j}a_{j}q_{j}v_{j}}{\sum_{i}v_{i}a_{i}+v_{j}}\left(q_{j}-\lambda^{*}\right)\right]+\sum_{j\in P_{2}}\left[\frac{v_{j}a_{j}q_{j}v_{j}}{\sum_{i}v_{i}a_{i}+v_{j}}\left(q_{j}-\lambda^{*}\right)\right]
\]
\[
\geq\frac{\underline{v}\underline{q}}{\sum_{i}v_{i}a_{i}+\underline{v}}\sum_{j\in P_{1}}\left[a_{j}v_{j}\left(q_{j}-\lambda^{*}\right)\right]+\frac{\underline{v}\underline{q}}{\sum_{i}v_{i}a_{i}+\underline{v}}\sum_{j\in P_{2}}\left[a_{j}v_{j}\left(q_{j}-\lambda^{*}\right)\right]
\]
\[
=\frac{\underline{v}\underline{q}}{\sum_{i}v_{i}a_{i}+\underline{v}}\sum_{j=1}^n\left[a_{j}v_{j}\left(q_{j}-\lambda^{*}\right)\right].
\]
By Equation \ref{eq:opt},
\[
\sum_{j=1}^{n}\left[a_{j}v_{j}\left(q_{j}-\lambda^{*}\right)\right]=0
\]
and the result follows.
\end{proof}

\noindent
\begin{cor}
  In expectation, the P-ranking policy under social influence achieves
  more downloads than the any ranking policy under the independent
  condition.
\end{cor}

\noindent
This corollary follows directly from the fact that the expected
downloads at each step do not change under the independent condition
regardless of the ranking policy in use.

\section{Discussion}
\label{section-conclusion}

This paper presented a ``measure and optimize'' algorithm for the
\musiclab{} which computes, at each refresh step, an optimal ranking
of the songs, given the appeal, estimated quality, current download
counts, and position visibility. This performance ranking, which
maximizes the expected number of downloads, can be computed in
strongly polynomial time, making the \mo{} algorithm highly
scalable. Our experimental results reveal two key findings:
\begin{enumerate}
\item The \mo{} algorithm leverages social influence to bring
  significant benefits in expected downloads.

\item Social influence and position bias both bring improvements in
  expected downloads and their effects are cumulative.
\end{enumerate}
Our theoretical results formally validate the first finding and prove
that the performance ranking under social influence always achieves
more expected downloads than any ranking policy under the independent 
condition. 

Our results shed an interesting and novel light about social influence
and the unpredictability it creates in cultural markets.  The
unpredictability coming from social influence is often presented as an
undesirable property of cultural markets. {\em However, as our
  experimental and theoretical results show, social influence, when
  used properly, is guaranteed to help in maximizing the efficiency of
  the market.} So, unless predictability is a goal per se, social
influence should be seen as an asset for optimizing cultural markets.

Our results also show that, in reasonable settings and with a
performance ranking, the unpredictability created by the use of social
influence is small. Moreover, the performance ranking identifies
``blockbusters'' and these ``blockbusters'' are the primary reason for
the increased efficiency of the market. There are conditions in the
\musiclab{} where the best song is not recovered and does not have the
most downloads. This happens when the best song has a ``bad'' appeal
and is dominated by another song. Assume, for instance, that we have
only two songs and $q_1 > q_2$. If $v_1 A_1 q_1 < v_2 A_2 q_2$ ($v_1 >
v_2$), the added visibility is not able to overcome the bad appeal:
Song 1, despite its better quality, will be dominated by song 2. Note
however that these situations can be predicted, since these dominance
relationships can be identified. Equally important, social influence
still helps in these settings compared to policies not using it.

The unpredictability when using the performance ranking comes from
songs that are indistinguishable, i.e., $ A_i q_i \approx A_j q_j$. In
this case, social influence may favor one of the songs depending on
how the process unfolds. It is interesting to point out once again
that these situations can be identified a priori. Together with the
dominance relationships, these equivalence relationships helps us
identify the sources of unpredictability, i.e., which songs may emerge
in the \musiclab{}. Moreover, it is conceivable that we can also
control this unpredictability to ensure fairness and reduce
inequalities. This is a topic we are currently investigating.

Our experiments have also shown that, contrary to the download
ranking, the performance ranking does not transform average songs into
``hits''. In reasonable settings, it always identifies the
``blockbusters'' while, in a worst-case scenario where song appeal is
negatively correlated with song quality, it still promotes
high-quality songs. There is a certain robustness in the performance
ranking which is not present in the download ranking.

A key insight from this research is the need to leverage social
influence carefully. The download ranking uses social influence twice:
Directly through the download counts and indirectly through the
ranking. The experimental results indicate that the resulting
algorithm puts too much emphasis on download counts, which are not
necessarily a good indication of quality due to the sampling
process. In contrast, the performance ranking balances appeal,
quality, and social influence. 

Observe that, in the generative model of the \musiclab{}, the
perceived quality of a song is not affected by social influence: Only
the probability of listening to the song is. Our results continue to
hold when the perceived quality of a song is affected by social
influence \cite{Lorenz2011,Farrell2011,Heiko2011}, e.g., when the
perceived quality of song $i$ is improved by $\epsilon_i \geq 0$, when
song $i$ is downloaded. Similarly, the results generalize to the case
where the sampling probability is given by
\[
p_{i,k}(\sigma) =  \frac{v_{\sigma_i}(\alpha A_i+f(D_{i,k}))}{\sum_{j=1}^n v_{\sigma_j}(\alpha A_j+f(D_{j,k}))},
\]
where function $f$ is positive nondecreasing. 

We hope that these results will revive the debate about the
consequences of social influence. Our results show that social
influence is beneficial for the efficiency of the market and that much
of its induced unpredictability can be controlled. It appears that a
key issue is to continue building our understanding of the nature of
social influence and to determine how best to use it and for which
purposes: Efficiency, predictability, fairness, ...

\section*{Acknowledgments}

We would like to thank Franco Berbeglia for some key contributions to
the proofs and Lorenzo Coviello for a very detailed reading of the
first draft of this paper. NICTA is funded by the Australian
Government as represented by the Department of Broadband,
Communications and the Digital Economy and the Australian Research
Council through the ICT Centre of Excellence program.

\bibliographystyle{plain} 
\bibliography{Music_Market_Manuel}

\end{document}